\documentclass[a4paper,UKenglish,cleveref, autoref, thm-restate]{lipics-v2021}

\hideLIPIcs  

\usepackage{amsfonts, amsmath, amssymb, amsthm}
\usepackage{todonotes}
\usepackage{xspace}
\usepackage{url}
\usepackage[short,nocomma]{optidef}

\newcommand{\cO}{\mathcal{O}}
\def\dd{\mathinner{.\,.}}

\usepackage[normalem]{ulem}
\usetikzlibrary{calc,decorations,shapes,arrows.meta}
\usepackage{booktabs}
\usepackage{bold-extra}

\def\dd{\mathinner{.\,.}}

\newcommand{\MM}{\textsc{Minimizing the Minimizers}\xspace}
\newcommand{\MMk}{\textsc{Minimizing the Minimizers} $(\leq\Sigma^k)$\xspace}
\newcommand{\MMONE}{\textsc{Minimizing the Minimizers} $(\leq\Sigma^1)$\xspace}
\newcommand{\MMD}{\textsc{Minimizing the Minimizers (Decision)}\xspace}
\newcommand{\FAS}{\textsc{Feedback Arc Set}\xspace}
\newcommand{\FASD}{\textsc{Feedback Arc Set (Decision)}\xspace}
\newcommand{\EFAS}{\textsc{Eulerian Feedback Arc Set}\xspace}
\newcommand{\EFASD}{\textsc{Eulerian Feedback Arc Set (Decision)}\xspace}

\newcommand{\ta}{\texttt{a}}
\newcommand{\tb}{\texttt{b}}
\newcommand{\Tab}{T_{\ta\tb}}
\newcommand{\Mab}{M_{\ta<\tb}}
\newcommand{\Mba}{M_{\tb<\ta}}
\newcommand{\Mwk}{\mathcal{M}_{w,k}}
\newcommand{\Mw}{\mathcal{M}_{w,1}}

\newcommand{\defproblem}[3]{
\vspace{2mm}
\noindent\fbox{
   \begin{minipage}{0.96\textwidth}
   \textsc{#1}\\
   {\bf{Input:}} #2  \\
   {\bf{Output:}} #3
   \end{minipage}
   }
   \vspace{2mm}
}

\newcommand{\DrawBracket}[4]{
	\draw  (#1, #3)
		-- (#1, #3 + 0.2em)
		-- node [anchor=south] {#4}
		   (#2, #3 + 0.2em)
		-- (#2, #3);
}

\newcommand{\DrawBracketUnder}[4]{
	\draw  (#1, #3)
		-- (#1, #3 - 0.2em)
		-- node [anchor=north] {#4}
		   (#2, #3 - 0.2em)
		-- (#2, #3);
}

\title{Minimizing the Minimizers via Alphabet Reordering} 

\author{Hilde Verbeek}{CWI, Amsterdam, The Netherlands}{hilde.verbeek@cwi.nl}{https://orcid.org/0000-0002-2399-3098}{Supported by a Constance van Eeden Fellowship.}
\author{Lorraine A.K. Ayad}{Brunel University London, London, UK}{lorraine.ayad@brunel.ac.uk}{https://orcid.org/0000-0003-0846-2616}{}
\author{Grigorios Loukides}{King's College London, London, UK}{grigorios.loukides@kcl.ac.uk}{ https://orcid.org/
0000-0003-0888-5061}{}
\author{Solon P. Pissis}{CWI, Amsterdam, The Netherlands \and Vrije Universiteit, Amsterdam, The Netherlands}{solon.pissis@cwi.nl}{https://orcid.org/0000-0002-1445-1932}{Supported by the PANGAIA and ALPACA projects that have received funding from the European Union’s Horizon 2020 research and innovation programme under the Marie Skłodowska-Curie grant agreements No 872539 and 956229, respectively.}

\authorrunning{H. Verbeek et al.} 

\Copyright{Hilde Verbeek, Lorraine A.K. Ayad, Grigorios Loukides, and Solon P. Pissis} 



\category{} 

\relatedversion{} 

\nolinenumbers 

\begin{document}

\maketitle

\begin{abstract}
Minimizers sampling is one of the most widely-used mechanisms for sampling strings [Roberts et al., Bioinformatics 2004].
Let $S=S[1]\ldots S[n]$ be a string over a totally ordered alphabet $\Sigma$.
Further let $w\geq 2$ and $k\geq 1$ be two integers.
The minimizer of $S[i\dd i+w+k-2]$ is the smallest position in $[i,i+w-1]$ where the lexicographically
smallest length-$k$ substring of $S[i\dd i+w+k-2]$ starts.
The set of minimizers over all $i\in[1,n-w-k+2]$ is the set $\Mwk(S)$ of the minimizers of $S$.

We consider the following basic problem: 
\begin{center}
 \emph{Given $S$, $w$, and $k$, can we efficiently compute a total order on $\Sigma$ that minimizes $|\Mwk(S)|$?}  
\end{center}

\noindent We show that this is unlikely by proving that the problem is NP-hard \emph{for any $w\geq 2$ and $k\geq 1$}. Our result provides theoretical justification as to why \emph{there exist no exact algorithms} for minimizing the minimizers samples, while \emph{there exists a plethora of heuristics} for the same purpose. 
\end{abstract}

\section{Introduction}

The minimizers sampling mechanism has been introduced independently by Schleimer et al.~\cite{DBLP:conf/sigmod/SchleimerWA03} and by Roberts et al.~\cite{DBLP:journals/bioinformatics/RobertsHHMY04}.
Since its inception, it has been employed ubiquitously in modern sequence analysis methods underlying some of the most widely-used tools~\cite{DBLP:journals/bioinformatics/Li16a,DBLP:journals/bioinformatics/Li18,Kraken}. 

Let $S=S[1]\ldots S[n]$ be a string over a totally ordered alphabet $\Sigma$.
Further let $w\geq 2$ and $k\geq 1$ be two integers.
The minimizer of the fragment $S[i\dd i+w+k-2]$ of $S$ is the smallest position in $[i,i+w-1]$ where the lexicographically
smallest length-$k$ substring of $S[i\dd i+w+k-2]$ starts.
We then define the set $\Mwk(S)$ of the minimizers of $S$ as
the set of the minimizers positions over all fragments $S[i\dd i+w+k-2]$, for $i\in[1,n-w-k+2]$.
Every fragment $S[i\dd i+w+k-2]$ containing $w$ length-$k$ fragments is called a \emph{window} of $S$.

\begin{example}\label{ex:first}
Let $S = \texttt{aacaaacgcta}$, $w = 3$, and $k=3$.
Assuming $\texttt{a}<\texttt{c}<\texttt{g}<\texttt{t}$, we have that $\Mwk(S)=\{1, 4, 5, 6, 7\}$. The minimizers positions are colored red: $S = \texttt{\textcolor{red}{a}ac\textcolor{red}{a}\textcolor{red}{a}\textcolor{red}{a}\textcolor{red}{c}gcta}$.
\end{example}

Note that by choosing
the \emph{smallest} position in $[i,i+w-1]$ where the lexicographically
smallest length-$k$ substring starts,
we resolve ties in case the latter substring has multiple occurrences in a window.

It is easy to prove that minimizers samples enjoy the following three useful properties~\cite{DBLP:journals/jcb/ZhengMK23}:
\begin{itemize}
    \item \textbf{Property 1 (approximately uniform sampling)}: Every fragment of length at least $w+k-1$ of $S$ has at least one representative position sampled by the mechanism.
    \item \textbf{Property 2 (local consistency)}: Exact matches between fragments of length at least $\ell \geq w+k-1$ of $S$ are preserved by means of having the same (relative) representative positions sampled by the mechanism.
    \item \textbf{Property 3 (left-to-right parsing)}: The minimizer selected by any fragment of length $w+k-1$ comes at or after the minimizers positions selected by all previous windows.
\end{itemize}

Since Properties $1$ to $3$ hold \emph{unconditionally},
and since the ordering of letters does not affect the correctness of algorithms using minimizers samples~\cite{DBLP:journals/spe/GrabowskiR17,DBLP:journals/almob/ShibuyaBK22,DBLP:journals/tkde/LoukidesPS23,DBLP:journals/pvldb/AyadLP23}, 
one would like to choose the ordering that minimizes the resulting sample as a means to improve the space occupied by the underlying data structures; contrast Example~\ref{ex:first} to the following example.

\begin{example}
Let $S = \texttt{aacaaacgcta}$, $w = 3$, and $k=3$.
Assuming $\texttt{c}<\texttt{a}<\texttt{g}<\texttt{t}$, we have that $\Mwk(S)=\{3, 6, 7\}$. The minimizers positions are colored red: $S = \texttt{aa\textcolor{red}{c}aa\textcolor{red}{a}\textcolor{red}{c}gcta}$. In fact, this ordering is a best solution in minimizing $|\Mwk(S)|$, together with the orderings $\texttt{c}<\texttt{g}<\texttt{t}<\texttt{a}$ and $\texttt{c}<\texttt{g}<\texttt{a}<\texttt{t}$, which both, as well, result in $|\Mwk(S)|=3$.
\end{example}

\subparagraph{Our Problem.} We next formalize the problem of computing a best such total order on $\Sigma$:

\defproblem{\MM}{A string $S \in \Sigma^n$ and two integers $w\geq 2$ and $k\geq 1$.}{A total order on $\Sigma$ that minimizes $|\Mwk(S)|$.}

\subparagraph{Motivation.} A lot of effort has been devoted by the bioinformatics community to designing practical algorithms for minimizing the resulting minimizers sample~\cite{DBLP:journals/jcb/ChikhiLJSM15,DBLP:journals/bioinformatics/DeorowiczKGD15,DBLP:conf/wabi/OrensteinPMSK16,DBLP:journals/bioinformatics/ZhengKM20,DBLP:journals/bioinformatics/JainRZCWKP20,DBLP:journals/bioinformatics/ZhengKM21,DBLP:journals/jcb/HoangZK22}. Most of these approaches  consider the space of all orderings on $\Sigma^k$ (the set of all possible length-$k$ strings on $\Sigma$) instead of the ones on $\Sigma$; and employ \emph{heuristics} to choose some ordering resulting in a small sample (see Section~\ref{sec:MMk} for a discussion).
To illustrate the impact of reordering on the number of minimizers, we considered two real-world datasets and measured the difference in the number of minimizers between the worst and best reordering, among those we could consider in a reasonable amount of time. The first dataset we considered is the complete genome of Escherichia coli str.~K-12 substr.~MG1655.
For selecting minimizers, we considered different orderings on $\Sigma^k$.
We thus mapped every length-$k$ substring to its lexicographic rank in $\{\texttt{A,C,G,T}\}^k$ (assuming $\texttt{A} < \texttt{C} < \texttt{G} < \texttt{T})$ 
constructing a new string $S$ over $[1,|\Sigma|^k]$.
We then computed $|\Mw(S)|$ for different values of $(w,k)$ and orderings on $[1,|\Sigma|^k]$. It should be clear that this corresponds to computing the size of $\Mwk$ for the original sequence over $\{\texttt{A,C,G,T}\}$. 
The second dataset is the complete genome of SARS-CoV-2 OL663976.1.
\cref{fig:Mws} shows the $\min$ and $\max$ values of the size of the obtained minimizers samples. The results in \cref{fig:Mws} clearly show the impact of alphabet reordering on $|\Mw(S)|$: the gap between the $\min$ and $\max$ is quite significant as in all cases we have $2\min< \max$.
Note that we had to terminate the exploration of the whole space of orderings 
when $2\min< \max$ was achieved; 
hence the presented gaps are not even the largest possible. 

\begin{figure}[ht]
    \centering
        \subfloat[][Complete genome of Escherichia coli.\label{fig:ECOLI}]
        {\includegraphics[width=0.49\textwidth]{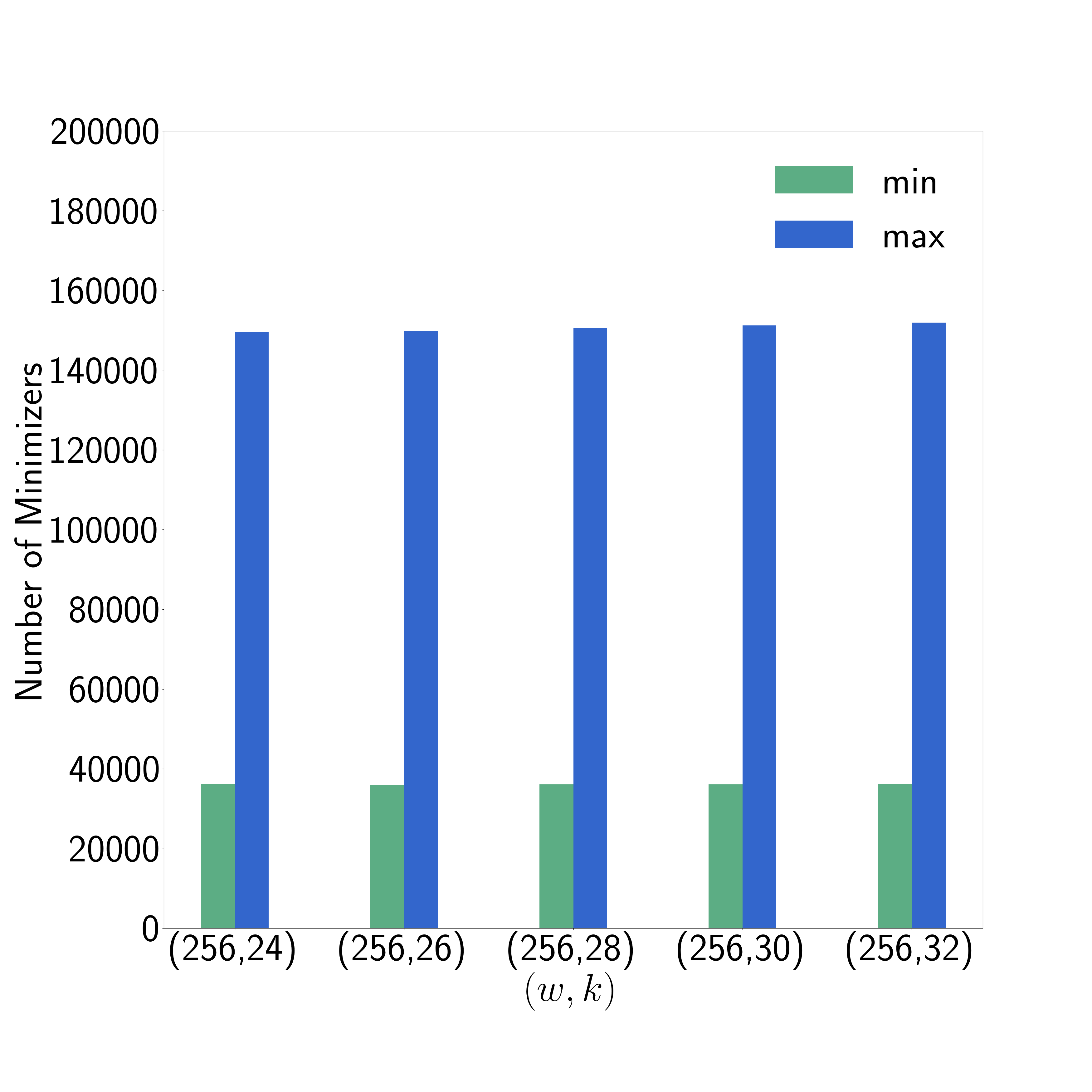}}\hspace{1mm}
        \subfloat[][Complete genome of SARS-CoV-2.\label{fig:SARS-CoV-2}]
        {\includegraphics[width=0.49\textwidth]{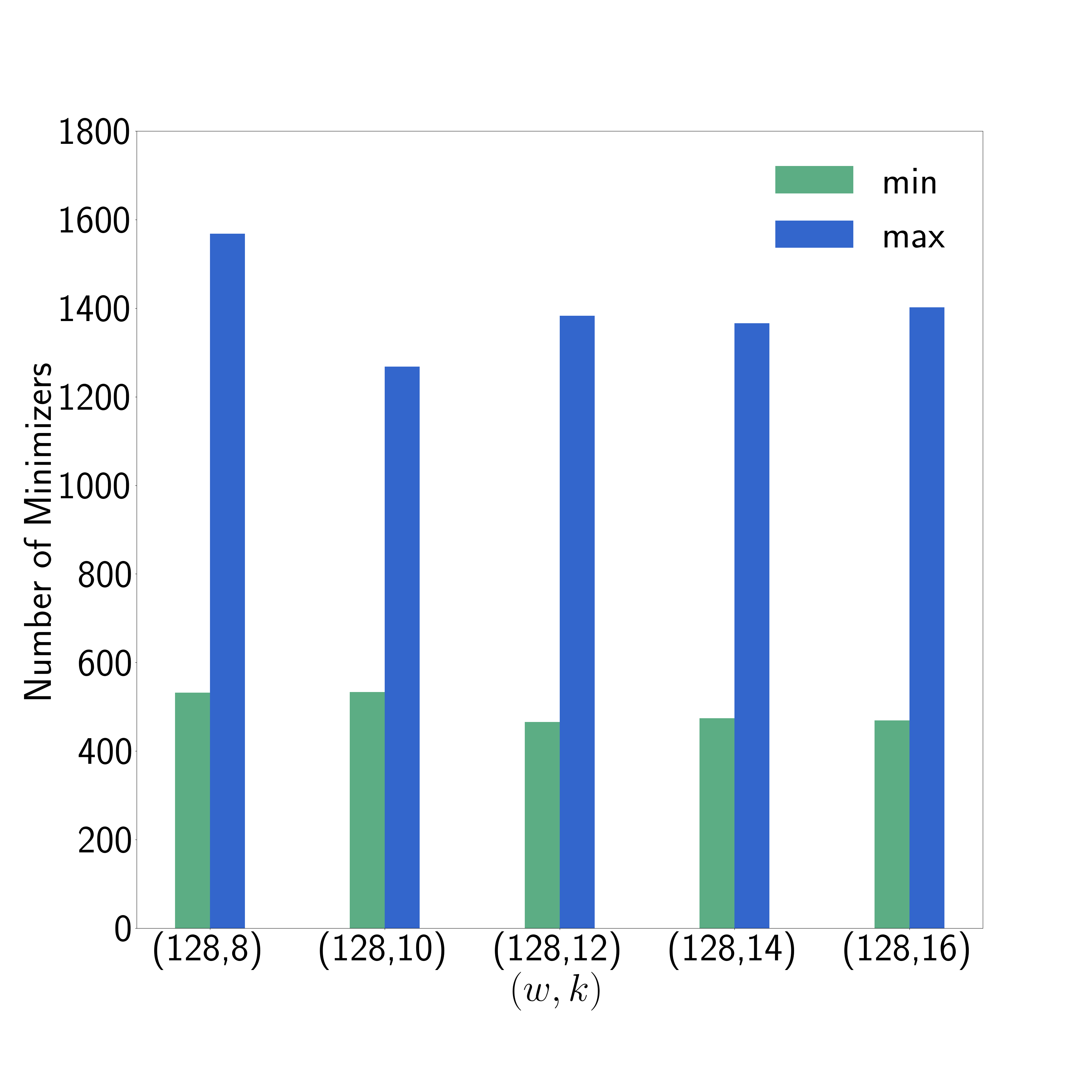}}
    \caption{The $\min$ and $\max$ values of the size of the minimizers sample, among \emph{some} of the possible orderings of $[1,|\Sigma|^k]$, on two real datasets using a range of $(w,k)$ parameter values.}
    \label{fig:Mws}
\end{figure}

This begs the question:

\begin{center}
 \emph{Given $S$, $w$, and $k$, can we efficiently compute a total order on $\Sigma$ that minimizes $|\Mwk(S)|$?}  
\end{center}

\subparagraph{Our Contribution.} We answer this basic question in the negative. Let us first define the decision version of \MM.

\defproblem{\MMD}{A string $S \in \Sigma^n$ and three integers $w\geq 2$, $k\geq 1$, and $\ell>0$.}{Is there a total order on $\Sigma$ such that $|\Mwk(S)|\leq \ell$?}

Our main contributions are the following two theorems implying \cref{cor:main}.
\begin{theorem}\label{the:main}
    \MMD is NP-complete if $w \ge 3$ and $k \ge 1$.
\end{theorem}
\begin{theorem}\label{the:main-w2}
    \MMD is NP-complete if $w = 2$ and $k \ge 1$.
\end{theorem}
\begin{corollary}\label{cor:main} \MMD is NP-complete for any $w$ and $k$.
\end{corollary}

\cref{cor:main} provides theoretical justification as to why \emph{there exist no exact algorithms} for minimizing the minimizers samples, while \emph{there exists a plethora of heuristics} for the same purpose. Notably, \cref{cor:main} settles the complexity landscape of the \MM problem. To cover \emph{all combinations} of input parameters $w$ and $k$ (i.e., for any $w \ge 2$ and $k \ge 1$), we design two non-trivial reductions from the feedback arc set problem~\cite{DBLP:conf/coco/Karp72}. In particular, the reduction in \cref{the:main} ($w \ge 3$) starts with any directed graph, whereas the reduction in \cref{the:main-w2} ($w = 2$) starts with any Eulerian directed graph.

The reductions we present are specifically for the case in which the size of the alphabet $\Sigma$ is variable. If $|\Sigma|$ is bounded by a constant, the problem can be solved in polynomial time: one can simply iterate over the $|\Sigma|!$ permutations of the alphabet, compute the number of minimizers for each ordering in linear time \cite{DBLP:conf/esa/LoukidesP21}, and output a globally best ordering.

\subparagraph{Other Related Work.} Choosing a best total order on $\Sigma$ is generally not new; it has also been investigated in other contexts, e.g., for choosing a best total order for minimizing the number of runs in the Burrows-Wheeler transform~\cite{DBLP:conf/esa/BentleyGT20}; for choosing a best total order for minimizing (or maximizing) the number of factors in a Lyndon factorization~\cite{DBLP:conf/stacs/GibneyT21}; or for choosing a best total order for minimizing the number of bidirectional string anchors~\cite{DBLP:journals/tkde/LoukidesPS23}.

\subparagraph{Paper Organization.} \cref{sec:proof} presents the proof of \cref{the:main} ($w \ge 3$). \cref{sec:proof-w2} presents the proof of \cref{the:main-w2} ($w = 2$).
\cref{sec:MMk} presents a discussion on orderings on $\Sigma^k$ in light of \cref{cor:main}. We conclude this paper with an open question in \cref{sec:fin}.

\section{\MM is NP-complete for \texorpdfstring{$w \ge 3$}{w >= 3}}\label{sec:proof}

We show that the \MM problem is NP-hard by a reduction from the well-known \textsc{Feedback Arc Set} problem~\cite{DBLP:conf/coco/Karp72}. 
Let us first formally define the latter problem. 

\defproblem{\FAS}{A directed graph $G=(V,A)$.}{A set $F \subseteq A$ of minimum size such that $(V, A \setminus F)$ contains no directed cycles.}

We call any such $F \subseteq A$ a \emph{feedback arc set}. 
The decision version of the \FAS problem is naturally defined as follows.

\defproblem{\FASD}{A directed graph $G=(V,A)$ and an integer $\ell'>0$.}{Is there a set $F \subseteq A$ such that $(V, A \setminus F)$ contains no directed cycles and $|F|\leq \ell'$?}

An equivalent way of phrasing this problem is to find an ordering on the set $V$ of the graph's vertices, such that the number of arcs $(u,v)$ with $u > v$ is minimal~\cite{cycles}. Then this is a topological ordering of the graph $(V, A \setminus F)$, and will be analogous to the alphabet ordering in the \MM problem; see~\cite{DBLP:journals/tkde/LoukidesPS23} for a similar application of this idea.\footnote{Our proof is more general and thus involved because it works for any values $w \ge 3$ and $k \ge 1$, whereas the reduction from~\cite{DBLP:journals/tkde/LoukidesPS23} works only for some fixed parameter values.} If \MM is then solved on the instance constructed by our reduction, producing a total order on $V$, taking all arcs $(u,v)$ with $u > v$ should produce a feedback arc set of minimum size, solving the original instance of the \textsc{Feedback Arc Set} problem.

\subsection{Overview of the Technique}
Given any instance $G = (V,A)$ of \FAS, we will construct a string $S$ over alphabet $V$ and of length polynomial in $|A|$. Specifically, we define string $S$ as follows:
\begin{equation*}
    S = \prod_{(\ta, \tb) \in A} \Tab^{q + 4},
\end{equation*}
\noindent where $\Tab$ is a string consisting of the letters $\ta$ and $\tb$, whose length depends only on $w$ and $k$, and $q$ is an integer polynomial in $|A|$, both of which will be defined later. The product $\prod$ of some strings is defined as their concatenation, and $X^q$ denotes $q$ concatenations of string $X$ starting with the empty string; e.g., if $X=\texttt{ab}$ and $q=4$, we have $X^q=(\texttt{ab})^4=\texttt{abababab}$.

String $\Tab$ will be designed such that each occurrence, referred to as a \emph{block}, will contain few minimizers if $\ta < \tb$ in the alphabet ordering, and many minimizers if $\tb < \ta$, analogous to the ``penalty'' of removing the arc $(\ta,\tb)$ as part of the feedback arc set. We denote by $\Mab$ the number of minimizers starting within some occurrence of $\Tab$ in $S$, provided that this $\Tab$ is both preceded and followed by at least two other occurrences of $\Tab$ (i.e., the middle $q$ blocks), when $\ta < \tb$ in the alphabet ordering. We respectively denote by $\Mba$ the number of minimizers starting in such a block when $\tb < \ta$ in the alphabet ordering. This will allow us (see \cref{fig:setting}) to express the total number of minimizers in $S$ in terms of $|F|$, the size of the feedback arc set, minus some discrepancy denoted by $\lambda$. This \emph{discrepancy} is determined by the blocks $\Tab$ that are not preceded or followed by two occurrences of $\Tab$ itself; namely, those that occur near some $T_{\tt{cd}}$, for another arc $(\tt{c},\tt{d})$, or those that occur near the start or the end of $S$.

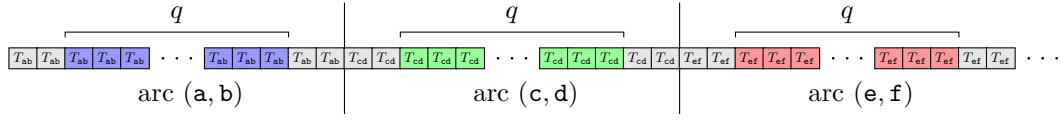
\begin{figure}[ht]
		\centering
		\begin{tikzpicture}
			\def\wid {0.7cm};
			\def\scale {0.525};
			
			\foreach \aa / \bb / \cold [count=\ix, evaluate=\ix as \xs using (12 * (\ix - 1) * \wid * \scale)] in {a/b/blue!40!white,	c/d/green!40!white, e/f/red!40!white} {
					\begin{scope}[xshift=\xs]
						\foreach \x in {0,...,11} {
							\def\col{\cold};
							\ifnum \x<2
								\def\col{black!10!white};
							\fi
							
							\ifnum \x>9
								\def\col{black!10!white};
							\fi
							
							\ifnum \x=5
							\else
								\ifnum \x=6
								\else
									\node [fill=\col, draw=black, minimum width=\wid, scale=\scale] at (\x * \wid * \scale,0) {$T_{\texttt{\aa\bb}}$};
								\fi
							\fi
						}
						
						\fill foreach \x in {0,1,2} {
							(5 * \wid * \scale + \x * 0.5 * \wid * \scale, 0) circle (0.5pt)
						};
						
						\DrawBracket{1.5 * \wid * \scale}{9.5 * \wid * \scale}{1.5em * \scale}{$q$}
						
						
						\node [anchor=north] at (6 * \wid * \scale, -1em * \scale) {arc $(\texttt{\aa},\texttt{\bb})$};
						
						\ifnum \ix<3
							\draw [thin]
								(11.5 * \wid * \scale, 4em * \scale) --
								(11.5 * \wid * \scale, -4em * \scale);
						\fi
					\end{scope}
			}
			
			\fill foreach \x in {0,1,2} {
				(36 * \wid * \scale + \x * 0.5 * \wid * \scale, 0) circle (0.5pt)
			};
		\end{tikzpicture}
		
		\caption{Illustration of the structure of string $S$, with the different gadgets for different arcs in $G$. The highlighted blocks are the ones for which the minimizers are counted in $\Mab$ and $\Mba$.}\label{fig:setting}
	\end{figure}

Let us start by showing an upper and a lower bound on the discrepancy $\lambda$.

\begin{lemma}
    $|A| - 1 \leq  \lambda \le 4 \cdot |A| \cdot |\Tab|$ if $|\Tab| \ge \frac{1}{4}(w + k - 1)$.
    \label{lem:lambda}
\end{lemma}
\begin{proof}
   We are counting the number of minimizers in $q$ blocks of $\Tab$, for each arc $(\ta, \tb)$. Note that we ignore four blocks for each arc, which is $4 \cdot |A|$ blocks of length $|\Tab|$ in total. This is $4 \cdot |A| \cdot |\Tab|$ positions in total, which gives the upper bound on the number of disregarded minimizers. For the lower bound, note that, by hypothesis, four consecutive blocks are at least as long as a single minimizer window, meaning at least one minimizer must be missed among the four blocks surrounding the border between each pair of consecutive arcs. The lower bound follows by the fact that for $|A|$ arcs we have $|A|-1$ such borders.
\end{proof}

Given the values $\Mab$, $\Mba$, and $\lambda$, we can express the total number of minimizers as a function of some feedback arc set $F$: if an arc $(\ta, \tb)$ is part of the feedback arc set, this corresponds to $\tb < \ta$ in the alphabet ordering, so the corresponding blocks $\Tab$ will each have $\Mba$ minimizers, whereas if $(\ta, \tb)$ is not in $F$, we have $\ta < \tb$ and the blocks will each have $\Mab$ minimizers. Using these values, we can define the number of minimizers in $S$ given some feedback arc set $F$ as
\begin{align}
    \Mwk(S, F) &= q \cdot \Mba \cdot |F| + q \cdot \Mab \cdot (|A| - |F|) + \lambda \nonumber\\
            &= q \cdot (\Mba - \Mab) \cdot |F| + q \cdot \Mab \cdot |A| + \lambda.\label{eq:mwk}
\end{align}

With this in mind, we can prove the following relationship between $\Mwk(S, F)$ and $|F|$:

\begin{lemma}
    Let $\ell'$ be some positive integer and let $\ell = q \cdot (\Mba - \Mab) \cdot (\ell' + 1) + q \cdot \Mab \cdot |A|$. If $\Mba > \Mab$, $|\Tab| \ge \frac{1}{4}(w + k - 1)$, and $q$ is chosen such that $\lambda < q \cdot (\Mba - \Mab)$, then $\Mwk(S, F) \le \ell$ if and only if $|F| \le \ell'$.
    \label{lem:mwk}
\end{lemma}

\begin{proof}
By hypothesis, $\Mba - \Mab$ is positive, thus, by Equation~\ref{eq:mwk}, $\Mwk(S,F)$ grows linearly with $|F|$. Suppose we have a feedback arc set $F$ with $|F| \le \ell'$. Consider the alphabet ordering inducing $F$ and let $\lambda$ be the corresponding discrepancy for $\Mwk(S, F)$. By hypothesis, we have $\lambda < q \cdot (\Mba - \Mab)$. Substituting the bounds on $|F|$ and $\lambda$ into Equation~\ref{eq:mwk} gives
\begin{align*}
    \Mwk(S, F) &\leq q \cdot (\Mba - \Mab) \cdot \ell' + q \cdot \Mab \cdot |A| + q \cdot (\Mba - \Mab) \\
               &= q \cdot (\Mba - \Mab) \cdot (\ell' + 1) + q \cdot \Mab \cdot |A| = \ell,
\end{align*}
\noindent completing the proof in one direction.

For the other direction, suppose we have picked $F$ such that $\Mwk(S, F) \le \ell$ and assume that $|F| \ge \ell' + 1$ towards a contradiction. Then we have the following two inequalities:
\begin{align*}
    \Mwk(S, F) &\le \ell = q \cdot (\Mba - \Mab) \cdot (\ell' + 1) + q \cdot \Mab \cdot |A| \\
    \Mwk(S, F) &\ge q \cdot (\Mba - \Mab) \cdot (\ell' + 1) + q \cdot \Mab \cdot |A| + \lambda.\hspace{+5mm}\text{(by Equation~\ref{eq:mwk})}
\end{align*}
By Lemma~\ref{lem:lambda}, for any non-trivial instance with $|A|>1$, $\lambda$ is strictly positive, meaning these inequalities are contradictory. Therefore, if $\Mwk(S, F) \le \ell$, it must be that $|F| \le \ell'$.%
\end{proof}

Given $w$ and $k$, we must determine a string $\Tab$ such that $\Mba > \Mab$ and $|\Tab| \ge \frac{1}{4}(w + k - 1)$. We then simply have to choose some $q$, which is polynomial in $|A|$, satisfying $\lambda < q \cdot (\Mba - \Mab)$. At that point we will have constructed a string $S$ for which it holds that the feedback arc set induced by the minimum set of minimizers is also a minimum feedback arc set on $G$, thus completing the reduction.

The following three subsections address the $\Tab$ construction:
\begin{itemize}
    \item \cref{sec:caseA}: $w \ge k + 2$ (Case A);
    \item \cref{sec:caseB}: $w = 3$ and $k \geq 2$ (Case B);
    \item \cref{sec:caseC}: $3 < w < k + 2$ (Case C).
\end{itemize}

It should be clear that the above sections cover all the cases for $w\geq 3$ and $k\geq 1$.
\cref{sec:complete} puts everything together to complete the proof.

\subsection{Case A: \texorpdfstring{$w \ge k + 2$}{w >= k + 2}} \label{sec:caseA}
\begin{lemma}
    Let $\Tab = \ta\tb^{w - 1}$, for $w \ge k + 2$. Then $\Mab = 1$ and $\Mba = w - k$.
    \label{lem:m-1}
\end{lemma}
\begin{proof}
    The block has length $w$; inspect \cref{fig:caseA}. Recall that, for the window starting at position $i$, the candidates for its minimizer are the length-$k$ fragments starting at positions $[i,i + w - 1]$. Therefore, for every window starting in a block $\Tab$ (provided it is succeeded by another $\Tab$), a candidate minimizer is $\ta\tb^{k-1}$; so if $\ta < \tb$, each $\Tab$ will contain just one minimizer. Thus we have $\Mab = 1$.
    
    For $\tb < \ta$, consider that $\Tab$ contains $w - k$ occurrences of $\tb^k$, and that for each window, at least one of the candidates for its minimizer is $\tb^k$. Since there is no length-$k$ substring that is lexicographically smaller than $\tb^k$, each occurrence of $\tb^k$ (and nothing else) is a minimizer, so it follows that $\Mba = w - k$. Note that $\Mba > \Mab$ only if $w \ge k + 2$.    
\end{proof}

\begin{figure}[ht]
		\centering
		\begin{tikzpicture}
			\def\gr{black!40!white}
			\def\patt{a,b,b,b,b,b,b};
			
			\foreach \ch [count=\i] in \patt {
				\node[anchor=south, color=\gr] at (\i em - 8em, 0) {\texttt{\ch}};
			}
			
			\foreach \ch [count=\i] in \patt {
				\node[anchor=south, color=\gr] at (\i em + 6em, 0) {\texttt{\ch}};
			}
			
			\foreach \ch [count=\i] in \patt {
				\node [anchor=south] at (\i em - 1em, 0) {\texttt{\ch}};
			}
			\draw [->] (0, 2.3em) -- (0, 1.4em);
			\draw [->] (1em, -1em) -- (1em, -.1em);
			\draw [->] (2em, -1em) -- (2em, -.1em);
			\draw [->] (3em, -1em) -- (3em, -.1em);
			
			\draw [color=\gr] (-0.5em, 2.3em) -- (-0.5em, -1em);
			\draw [color=\gr] (6.5em, 2.3em) -- (6.5em, -1em);
			
			\DrawBracket{2.75em}{6.25em}{1.4em}{$k$}
		\end{tikzpicture}
		
		\caption{Illustration of $3$ copies of $\Tab$ in $S$ for $w = 7$ and $k = 4$, along with its respective minimizers when $\ta < \tb$ (top) and when $\tb < \ta$ (bottom). It can be seen that $\Mab = 1$ and $\Mba = 3$.}
  \label{fig:caseA}
	\end{figure}

\subsection{Case B: \texorpdfstring{$w = 3$ and $k \geq 2$}{w = 3 and k >= 2}}\label{sec:caseB}
\begin{lemma}
    Let $\Tab = (\ta\tb)^t \tb\tb$ with $t = \left\lceil \frac{w + k}{2} \right\rceil$, for $w = 3$ and $k \geq 2$. Then $\Mab = \left\lfloor \frac{k}{2} \right\rfloor + 3$ and $\Mba = \left\lfloor \frac{k}{2} \right\rfloor + 4$.
    \label{lem:m-2}
\end{lemma}
\begin{proof}
Since $w=3$, for every window, the minimizer is one out of three length-$k$ fragments; inspect \cref{fig:caseB}. Every $\ta$ in the block has a $\tb$ before it. For any window starting at a position preceding an $\ta$, two of the candidates start with a $\tb$ and the other starts with an $\ta$. 
As an example consider the window $\texttt{babab}$ preceding an $\ta$ in \cref{fig:caseB}.
We have that the first and the third candidates start with a $\tb$ and the second starts with an $\ta$.
Therefore, if $\ta < \tb$, the candidate starting with an $\ta$ will be chosen and every $\ta$ in $\Tab$ is a minimizer. Only the window starting at the third-to-last position of the block will not consider any length-$k$ substring starting with an $\ta$ as its minimizer, as therein we have three $\tb$'s occurring in a row. Since $k \geq 2$, the last $\tb$ of the block will be chosen if $\ta < \tb$. Thus, $\Mab$ counts every $\ta$ and one $\tb$, which gives:
\begin{equation*}
    \Mab = t + 1 = \left\lceil \frac{w + k}{2} \right\rceil + 1 = 
    \left\lceil \frac{3+k}{2} \right\rceil + 1 = \left\lfloor \frac{k+2}{2} \right\rfloor + 1+1=\left\lfloor \frac{k}{2} \right\rfloor + 3.
\end{equation*}
For $\Mba$, we apply the same logic to conclude that every $\tb$ surrounded by $\ta$'s is a minimizer, which accounts for all $\tb$'s except the final three, which occur at positions $[2t,2t + 2]$:

\begin{itemize}
    \item For the window starting at position $2t$, the three minimizer candidates start, respectively, 
    with $\tb\tb$, $\tb\tb$ and $\tb\ta$. 
    Since $k \ge 2$, the first candidate ($2t$) will be the minimizer because it is lexicographically a smallest and the leftmost ($\texttt{b}<\texttt{a}$).
    \item For the window starting at position $2t + 1$, the first two candidates start, respectively, 
    with $\tb\tb$ and $\tb\ta$, and the third starts with an $\ta$.
    The first candidate ($2t + 1$) will be the minimizer, because it is lexicographically smaller ($\texttt{b}<\texttt{a}$).
    \item For the window starting at position $2t + 2$, the first and third candidates start with a $\tb$ whereas the second starts with an $\ta$. The third candidate starts at the second position of the next $\Tab$-block. Since $2t > k + 1$, this candidate consists of only $\tb\ta\tb\ta\dots$ alternating for $k$ letters. It is equal to the first candidate, so by tie-breaking the first candidate ($2t + 2$) is the minimizer as it is the leftmost.
\end{itemize}

Thus, every $\tb$ in the block will be a minimizer if $\tb < \ta$, and we have:

\begin{equation*}
    \Mba = t + 2 = \left\lceil \frac{3+k}{2} \right\rceil + 2 = \left\lfloor \frac{k}{2} \right\rfloor + 4.
\end{equation*}
\end{proof}

	\begin{figure}[ht]
		\centering
		\begin{tikzpicture}
			\def\gr{black!40!white}
			\def\patt{a,b,a,b,a,b,b,b};
			
			\foreach \ch [count=\i] in \patt {
				\node[anchor=south, color=\gr] at (\i em - 9em, 0) {\texttt{\ch}};
			}
			
			\foreach \ch [count=\i] in \patt {
				\node[anchor=south, color=\gr] at (\i em + 7em, 0) {\texttt{\ch}};
			}
			
			\foreach \ch [count=\i] in \patt {
				\node [anchor=south] at (\i em - 1em, 0) {\texttt{\ch}};
			}
			\draw [->] (0, 2.3em) -- (0, 1.4em);
			\draw [->] (2em, 2.3em) -- (2em, 1.4em);
			\draw [->] (4em, 2.3em) -- (4em, 1.4em);
			\draw [->] (7em, 2.3em) -- (7em, 1.4em);
			\draw [->] (1em, -1em) -- (1em, -.1em);
			\draw [->] (3em, -1em) -- (3em, -.1em);
			\draw [->] (5em, -1em) -- (5em, -.1em);
			\draw [->] (6em, -1em) -- (6em, -.1em);
			\draw [->] (7em, -1em) -- (7em, -.1em);
			
			\draw [color=\gr] (-0.5em, 2.3em) -- (-0.5em, -1em);
			\draw [color=\gr] (7.5em, 2.3em) -- (7.5em, -1em);
			
		\end{tikzpicture}
		
		\caption{$\Tab$ for $w = 3$ and $k = 3$, with its respective minimizers. The last $\tb$ is a minimizer even when $\ta < \tb$, because $w = 3$. In this situation, $\Mab = 4$ and $\Mba = 5$.}
        \label{fig:caseB}
	\end{figure}

\subsection{Case C: \texorpdfstring{$3 < w < k + 2$}{3 < w < k + 2}} \label{sec:caseC}
\begin{lemma}
    Let $\Tab = (\ta\tb)^t \tb\tb$ with $t = \left\lceil \frac{w + k}{2} \right\rceil$, for $3 < w < k + 2$. Then 
    \begin{itemize}
        \item if $k$ is even, $\Mab = \frac{k}{2} + 2 + p$ and $\Mba = \frac{k}{2} + 3 + p$, where $p = (w + k) \mod 2$;
        \item if $k$ is odd, $\Mab = \left\lfloor \frac{k}{2}\right\rfloor + 3$ and $\Mba = \left\lfloor \frac{k}{2}\right\rfloor + 4$.
    \end{itemize}
    \label{lem:m-3}
\end{lemma}
\begin{proof}
Every length-$w$ fragment of the block contains at least one $\ta$ and at least one $\tb$; inspect \cref{fig:caseC}. Because of this, only $\ta$'s will be minimizers if $\ta < \tb$ and only $\tb$'s if $\tb < \ta$ (unlike when $w = 3$, as shown in \cref{sec:caseB}). We start by counting $\Mab$. Suppose we are determining the minimizer at position $i$. Every candidate we consider is a string of alternating $\ta$'s and $\tb$'s (starting with an $\ta$), in which potentially one $\ta$ is substituted by a $\tb$ (if the length-$k$ fragment contains the $\tb\tb\tb$ at the end of the block). A lexicographically smallest length-$k$ fragment is one in which this extra $\tb$ appears the latest, or not at all.

First, we will consider the number of length-$k$ fragments in which the extra $\tb$ does not occur. For these fragments, it is the case that no other fragment in the block is lexicographically smaller when $\ta < \tb$, so it is automatically picked as minimizer at the position corresponding to the start of the length-$k$ fragment. The extra $\tb$ appears at position $2t + 1$ in the block, so this applies to all length-$k$ fragments starting with an $\ta$ that end before position $2t + 1$. That is, all $\ta$'s up to (and including) position $i = 2t - (k - 1) = 2\left\lceil \frac{w + k}{2} \right\rceil - k + 1 = w + k + p - k + 1 = w + p + 1$, where $p = (w + k) \mod 2$.

Next, we consider the length-$k$ fragments that do include the extra $\tb$. At any position past $i$, the smallest candidate will be the first one starting with an $\ta$, \emph{unless} one of the candidates appears in the next $\Tab$-block, in which case the minimizer will be the first position of this next block (because this candidate does not include the extra $\tb$ and is therefore smaller than any candidate before it). Specifically, this is the case if position $|\Tab| + 1$ is one of the $w$ candidates. Therefore, all windows starting at positions up to and including $j = |\Tab| + 1 - w =(2\lceil \frac{w+k}{2}\rceil+2)+1-w= w + k + 3 + p - w = k + p + 3$ will have as their minimizer the first position with an $\ta$, meaning that all $\ta$'s up to position $j + 1$ are minimizers.

\begin{figure}[ht]
	\centering
	\begin{tikzpicture}
		\def\gr{black!40!white}
		\def\patt{a,b,a,b,a,b,a,b,b,b};
		
		\foreach \ch [count=\i] in \patt {
			\node[anchor=south, color=\gr] at (\i em - 11em, 0) {\texttt{\ch}};
		}
		
		\foreach \ch [count=\i] in \patt {
			\node[anchor=south, color=\gr] at (\i em + 9em, 0) {\texttt{\ch}};
		}
		
		\foreach \ch [count=\i] in \patt {
			\node [anchor=south] at (\i em - 1em, 0) {\texttt{\ch}};
		}
		
		\draw [color=\gr] (-0.5em, 2.3em) -- (-0.5em, -1em);
		\draw [color=\gr] (9.5em, 2.3em) -- (9.5em, -1em);
		
		\DrawBracket{4.75em}{8.25em}{1.4em}{$k$}
		\DrawBracketUnder{6.75em}{10.25em}{-.4em}{$w$}
		
		\node[anchor=south] at (4em, 1.6em) {$i$};
		\draw [-|] (-.5em, 1.45em) -- (4.25em, 1.45em);
		\node[anchor=north] at (6em, -.2em) {$j$};
		\draw [-|] (-.5em, -.15em) -- (7.25em, -.15em);
	\end{tikzpicture}
	
	\caption{$\Tab$ for $w = 4$ and $k = 4$, showing the positions $i$ and $j$ for counting $\Mab$. Position $i$ is the final position at which the length-$k$ fragment does not contain $\tb\tb$, whereas $j$ is the final position for which the starting position of the next $\Tab$-block is not a candidate. When $\ta < \tb$, the minimizers in the block are all $\ta$'s up to position $\max \{i,j + 1\}$.}
    \label{fig:caseC}
	\end{figure}

We now have that all $\ta$'s up to position $i = w + p + 1$ and all $\ta$'s up to position $j + 1 = k + p + 4$ are minimizers. Thus we need to count the $\ta$'s up to position $\max \{w + p + 1, k + p + 4\}$. Because, by hypothesis, $w < k + 2$, this maximum is equal to $k + p + 4$. The first $k + p + 4$ letters of the block are alternating $\ta$'s and $\tb$'s, so we get

\begin{equation*}
    \Mab = \left\lceil \frac{k + p + 4}{2} \right\rceil = \left\lceil \frac{k + p}{2} \right\rceil + 2 = \begin{cases}
        \frac{k}{2} + 2 + p & \text{ if } k \text{ is even;} \\
        \left\lfloor \frac{k}{2}\right\rfloor + 3 & \text{ if } k \text{ is odd.}
    \end{cases}
\end{equation*}

Next, we compute $\Mba$. We start by showing that the final three $\tb$'s in $\Tab$ are all minimizers. There is only one length-$k$ fragment that starts with $\tb\tb\tb$ and one that starts with $\tb\tb\ta$, so the first two of these final $\tb$'s will both be minimizers for the windows that start with $\tb\tb\tb$ and $\tb\tb\ta$. For the window that starts at the third $\tb$, which is position $|\Tab|$, note that the entire window does not contain $\tb\tb$ at all; it consists of only alternating $\tb$'s and $\ta$'s as the window has length $w + k - 1$ whereas the next occurrence of $\tb\tb$ is after $w + k + p$ positions. Because the window does not contain $\tb\tb$, none of its candidates are smaller than $\tb\ta\tb\ta\dots$ alternating, which first appears at the start of the window. Therefore, the third $\tb$ is also a minimizer.

The rest of the minimizers consist of two sets. The first set corresponds to positions for which no candidate is smaller than $\tb\ta\tb\ta\dots$ (alternating for $k$ letters). These are all positions with a $\tb$, up to a certain position $i$ (to be computed later), after which there will also be a smaller minimizer candidate, i.e., one that contains $\tb\tb$; inspect~\cref{fig:caseCb}. This is the second set of minimizers: ones that start with $\tb$ and contain $\tb\tb$ at some point. These are all positions with a $\tb$ from some position $j$ onwards.

\begin{figure}[ht]
		\centering
		\begin{tikzpicture}
			\def\gr{black!40!white}
			\def\patt{a,b,a,b,a,b,a,b,b,b};
			
			\foreach \ch [count=\i] in \patt {
				\node[anchor=south, color=\gr] at (\i em - 11em, 0) {\texttt{\ch}};
			}
			
			\foreach \ch [count=\i] in \patt {
				\node[anchor=south, color=\gr] at (\i em + 9em, 0) {\texttt{\ch}};
			}
			
			\foreach \ch [count=\i] in \patt {
				\node [anchor=south] at (\i em - 1em, 0) {\texttt{\ch}};
			}
			
			\draw [color=\gr] (-0.5em, 2.3em) -- (-0.5em, -1em);
			\draw [color=\gr] (9.5em, 2.3em) -- (9.5em, -1em);
			
			\DrawBracket{4.75em}{8.25em}{1.4em}{$k$}
			\DrawBracket{0.75em}{4.25em}{1.4em}{$w$}
			
			\node[anchor=north] at (1em, -.2em) {$i$};
			\draw [-|] (-.5em, -.15em) -- (2.25em, -.15em);
			
			\node[anchor=north] at (5em, -.2em) {$j$};
			\draw [|-] (4.75em, -.15em) -- (9.5em, -.15em);
		\end{tikzpicture}
		
		\caption{$\Tab$ for $w = 4$ and $k = 4$, showing the positions $i$ and $j$ when counting $\Mba$: $j$ is the position of the first $\tb$ at which the corresponding length-$k$ fragment contains $\tb\tb$; $i$ is the last position at which $j$ is not a candidate for its minimizer. When $\tb < \ta$, the minimizers in this block are all $\tb$'s up to position $i + 1$ and all $\tb$'s from position $j$ onwards.}
        \label{fig:caseCb}
	\end{figure}

We start by computing $j$. Position $j$ is the first position such that the length-$k$ fragment starting at $j$ starts with a $\tb$ and contains $\tb\tb$. If $k$ is odd, the fragment ends at position $2t + 2$ with $\tb\tb\tb$ as suffix; if $k$ is even, the fragment ends at position $2t + 1$ with $\tb\tb$ as suffix. We have

\begin{equation*}
    j = \begin{cases}
        2t + 1 - k + 1 = w + p + 2 & \text{ if } k \text{  is even;} \\
        2t + 2 - k + 1 = w + p + 3 & \text{ if } k \text{ is odd.}
    \end{cases}
\end{equation*}
Note that $j = w + p + 2 + (k \mod 2)$. Every $\tb$ from position $j$ onwards is a minimizer. This includes the three $\tb$'s at the end of the pattern (at positions $2t$ through $2t + 2$), as well as the ones between positions $j$ and $2t - 1$ (both inclusive). Thus we have

\begin{align*}
         3 + \left\lfloor \frac{2t - j}{2} \right\rfloor &= 3 + \left\lfloor \frac{w + k + p - (w + p + 2 + (k \mod 2))}{2} \right\rfloor \\
    = \; 3 + \left\lfloor \frac{k - 2 - (k \mod 2)}{2} \right\rfloor &= 2 + \left\lfloor \frac{k}{2} \right\rfloor \\
\end{align*}

\noindent $\tb$'s from position $j$ onwards.

Next, we compute $i$ and count the number of $\tb$'s up to $i$. We take the last position for which the length-$k$ fragment starting at $j$ is not a candidate. This is $i = j - w$. The minimizer for the window starting at position $i + 1$ is the length-$k$ fragment starting at $j$, since this is the only candidate that contains $\tb\tb$. However, if there is a $\tb$ at position $i + 1$,\footnote{Consider the case when $\Tab=\texttt{abababababbb}$ with $w=5$ and $k=4$. For this block, we have $i=3$ and $j=8$. Indeed $i=j-w=3$ and at position $i+1=4$ of the block we have a $\tb$. Position $4$ will be selected as the minimizer for the window starting at position $3$.} then $i + 1$ will still be a minimizer: when we take the minimizer for position $i$, the length-$k$ fragment containing $\tb\tb$ will not be a candidate so it will take the first length-$k$ fragment starting with a $\tb$, which is at position $i + 1$. Therefore, we  count all $\tb$'s that appear \emph{up to} $i + 1$:

\begin{align*}\hspace{-3mm}
   \left\lfloor \frac{i + 1}{2} \right\rfloor =& \left\lfloor \frac{j - w + 1}{2} \right\rfloor = \left\lfloor \frac{(w+p+2+(k\mod 2)) - w + 1}{2} \right\rfloor=\left\lfloor \frac{p + 3 + (k \mod 2)}{2} \right\rfloor \\
    = &1 + \left\lfloor \frac{1 + p + (k \mod 2)}{2} \right\rfloor = \begin{cases}
        1 + p & \text{ if } k \text{ is even;} \\
        2     & \text{ if } k \text{ is odd.}
    \end{cases}
\end{align*}

Adding the two numbers of $\tb$'s together gives (inspect \cref{fig:caseCc}):

\begin{align*}
    \Mba &= 2 + \left\lfloor \frac{k}{2} \right\rfloor + \begin{cases}
        1 + p & \text{ if } k \text{ is even;} \\
        2     & \text{ if } k \text{ is odd;}
    \end{cases} \\
    &= \begin{cases}
        \frac{k}{2} + 3 + p & \text{ if } k \text{ is even;} \\
        \left\lfloor \frac{k}{2} \right\rfloor + 4     & \text{ if } k \text{ is odd.}
    \end{cases}
\end{align*}
\end{proof}

\begin{figure}[ht]
    \centering
    \begin{tikzpicture}
        \def\gr{black!40!white}
        \def\patt{a,b,a,b,a,b,a,b,b,b};
        
        \foreach \ch [count=\i] in \patt {
            \node[anchor=south, color=\gr] at (\i em - 11em, 0) {\texttt{\ch}};
        }
        
        \foreach \ch [count=\i] in \patt {
            \node[anchor=south, color=\gr] at (\i em + 9em, 0) {\texttt{\ch}};
        }
        
        \foreach \ch [count=\i] in \patt {
            \node [anchor=south] at (\i em - 1em, 0) {\texttt{\ch}};
        }
        
        \draw [color=\gr] (-0.5em, 2.3em) -- (-0.5em, -1em);
        \draw [color=\gr] (9.5em, 2.3em) -- (9.5em, -1em);
        
        \draw [->] (0, 2.3em) -- (0, 1.4em);
        \draw [->] (2em, 2.3em) -- (2em, 1.4em);
        \draw [->] (4em, 2.3em) -- (4em, 1.4em);
        \draw [->] (6em, 2.3em) -- (6em, 1.4em);
        \draw [->] (1em, -1em) -- (1em, -.1em);
        \draw [->] (5em, -1em) -- (5em, -.1em);
        \draw [->] (7em, -1em) -- (7em, -.1em);
        \draw [->] (8em, -1em) -- (8em, -.1em);
        \draw [->] (9em, -1em) -- (9em, -.1em);
    \end{tikzpicture}
    
    \caption{$\Tab$ for $w = 4$ and $k = 4$, showing its minimizers for $\ta < \tb$ (top) and $\tb < \ta$ (bottom). In this situation, $\Mab = 4$ and $\Mba = 5$.}
    \label{fig:caseCc}
\end{figure}
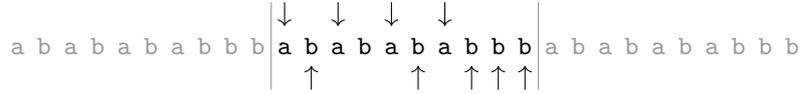

\subsection{Wrapping up the Reduction}\label{sec:complete}
\begin{proof}[Proof of Theorem~\ref{the:main}]
$\MMD$ asks whether or not there exists some ordering on $\Sigma$ such that a string $S \in \Sigma^n$ has at most $\ell$ minimizers for parameters $w$ and $k$. Given $w$, $k$ and an ordering on $\Sigma$, one can compute the number of minimizers for those parameters in linear time~\cite[Theorem 3]{DBLP:conf/esa/LoukidesP21}. Therefore, one can use an alphabet ordering as a certificate to verify a \textsf{YES} instance of $\MMD$ simply by comparing the computed number of minimizers to $\ell$. This proves that the $\MMD$ problem is in NP.
To prove that \MMD is NP-hard, we use a reduction from \FASD (see \cref{sec:proof} for definition), which is a well-known NP-complete problem~\cite{DBLP:conf/coco/Karp72}.

We are given an instance $G=(V,A)$ of \FAS and an integer $\ell'$, and we are asked to check if $G$ contains a feedback arc set with at most $\ell'$ arcs. We will construct an instance $S$ of \MMD, for given parameters $w \ge 3$ and $k \ge 1$, such that: the minimum number of minimizers in $S$, over all alphabet orderings, is at most some value $\ell$ if and only if $G$ contains a feedback arc set of size at most $\ell'$.

By Lemma~\ref{lem:lambda}, we have $\lambda \le 4 \cdot |A| \cdot |\Tab|$. Given $w$ and $k$, we must determine a string $\Tab$ such that $\Mba > \Mab$ and $|\Tab| \ge \frac{1}{4}(w + k - 1)$, and also choose some $q$ satisfying $\lambda < q \cdot (\Mba - \Mab)$ (see Lemma~\ref{lem:mwk}). Let $\Sigma = V$ and let $S = \prod_{(\ta, \tb) \in A}\Tab^{q + 4}$, with $\Tab$ and $q$ to be determined depending on $w$ and $k$.

\subparagraph{Case A: $w \ge k + 2.$} Let $\Tab = \ta\tb^{w - 1}$, so $|\Tab| = w$. 
Since, by hypothesis, the maximal value of $k$ is $w-2$, and since $|\Tab| = w$, we have that $4|\Tab| \geq 2w -3$.  
Thus, the condition on the length of $\Tab$ always holds. 
By Lemma~\ref{lem:m-1}, $\Mba - \Mab = w - k - 1$. We choose $q = 4 \cdot w \cdot |A| + 1$, so that $\lambda \le 4 \cdot |A|\cdot w < q \cdot (w - k - 1)$. Thus, $\lambda < q\cdot (\Mba-\Mab)$.

\subparagraph{Case B and Case C: $w < k + 2.$} Let $\Tab = (\ta\tb)^t\tb\tb$ for $t = \left\lceil \frac{w + k}{2} \right\rceil$. We have $|\Tab| = 2t + 2 = 2(\left\lceil \frac{w + k}{2} \right\rceil) + 2 = w + k + p + 2$, where $p = (w + k) \mod 2$. The condition on the length of $\Tab$ always holds because $w+k+p+2>w+k-1$.

\begin{itemize}
    \item If $w = 3$, then by Lemma~\ref{lem:m-2}, $\Mba - \Mab = \left\lfloor \frac{k}{2} \right\rfloor + 4 - (\left\lfloor \frac{k}{2} \right\rfloor + 3) = 1$.
    \item If $w > 3$, then by Lemma~\ref{lem:m-3}:
    \begin{itemize}
        \item if $k$ is even, $\Mba - \Mab = \frac{k}{2} + 3 + p - (\frac{k}{2} + 2 + p) = 1$;
        \item if $k$ is odd, $\Mba - \Mab = \left\lfloor \frac{k}{2} \right\rfloor + 4 - (\left\lfloor \frac{k}{2} \right\rfloor + 3) = 1$.
    \end{itemize}
\end{itemize}

In any case, $\Mba - \Mab = 1$. We choose $q = 4 \cdot |A| \cdot (w + k + 3) + 1$, so that $\lambda \le 4 \cdot |A| \cdot (w + k + p + 2) < q$. Thus, $\lambda < q\cdot (\Mba-\Mab)$. \medskip

Finally, we set $\ell = q \cdot (\Mba - \Mab) \cdot (\ell' + 1) + q \cdot \Mab \cdot |A|$. By Lemma~\ref{lem:mwk}, we have that $\Mwk(S, F) \le \ell$ if and only if $|F| \le \ell'$; in other words, $G$ contains a feedback arc set of size at most $\ell'$ if and only if $S$ has an alphabet ordering with at most $\ell$ minimizers.

Hence we have shown that $(G, \ell')$ is a \textsf{YES} instance of \MMD if and only if $(S, \ell)$ is a \textsf{YES} instance of \FASD. Moreover, the length of $S$ is $(q + 4) \cdot |A| \cdot |\Tab|$, with $\Tab$ being of polynomial length, so the reduction can be performed in polynomial time. The existence of a polynomial-time reduction from \FASD to \MMD proves our claim: \MMD is NP-complete if $w \ge 3$ and $k \ge 1$.
\end{proof}

\section{\MM is NP-complete for \texorpdfstring{$w = 2$}{w = 2}}\label{sec:proof-w2}

A weakly connected directed graph is called \emph{Eulerian} if it contains an Eulerian circuit: a trail which starts and ends at the same graph vertex; i.e., a graph cycle which uses each arc exactly once. By Euler's famous theorem, we know that a weakly connected directed graph is Eulerian if and only if every graph vertex has equal in-degree and out-degree. In the special case when $G(V,A)$ is a directed Eulerian graph, we call the problems \FAS and \FASD, \EFAS and \EFASD, respectively. In this section, we design a reduction from the \EFASD problem to the \MMD problem for the case of $w=2$. \EFASD is known to be NP-complete \cite{DBLP:journals/corr/abs-1303-3708}.

Given any directed Eulerian graph $G = (V,A)$, we construct a string $S$ on the alphabet $\Sigma = V$ such that the minimum number of minimizers $\Mwk(S, F)$ for $w=2$ can be expressed exactly as a linear  function of the cardinality of the feedback arc set $F$. We do this by constructing a string gadget for each arc $(\ta, \tb)\in A$, for which we can determine the number of minimizers if $\ta < \tb$ and if $\tb < \ta$. The fact that $G$ is Eulerian is used to allow these gadgets to \emph{overlap}, so that every minimizer window can correspond to a single arc or vertex.

\subparagraph{Construction of $S$.} Let $v_0,\dots,v_{|A|}$ denote an Eulerian circuit on the vertices of $G$. Every consecutive vertex pair $(v_i, v_{i + 1})$, $i\in[0,|A|)$, corresponds to a single arc in $A$, and we can have no more than $|A|$ such pairs. Note however that vertices may occur multiple times in the circuit. We construct $S$ as follows:

\begin{equation}\label{eq:construction-w2}
    S = \prod_{i=0}^{|A|} v_{|A| - i}^{k + 1}.
\end{equation}

Recall that the product $\prod$ of some strings denotes their concatenation, and raising a letter to some power $x$ means repeating it $x$ times. In other words, $S$ is created by taking the Eulerian circuit in \emph{reverse}, and repeating each letter $v_i$ $k + 1$ times. We call each $v_i^{k + 1}$ a \emph{block}.
By having a block of $v_{i+1}$ followed by a block of $v_i$ in $S$, we ensure that we have \emph{one extra minimizer} when $v_{i+1}<v_i$ and several other minimizers regardless of the ordering. We reverse the arcs to ensure that this extra minimizer is added to $\Mwk(S, F)$ when $(v_{i},v_{i+1})\in F$.

\begin{lemma}\label{lem:mwk-w2}
    $\Mwk(S, F) = 1 + |A| \cdot k + |F|$ if $w=2$.
\end{lemma}
\begin{proof}
    Since $w = 2$, the windows for which minimizers are determined have length $w + k - 1 = k + 1$. By construction, for any arc $(v_i, v_{i + 1})$, there are $k + 1$ length-$(k + 1)$ windows starting at some position in each block corresponding to $v_{i + 1}$. We consider each of these windows.
    
    The first window, starting at the first $v_{i + 1}$, is $v_{i + 1}^{k + 1}$. Since both candidates in this window are equal (to $v_{i + 1}^k$), the leftmost fragment will be the minimizer by tie-breaking regardless of whether or not $(v_i, v_{i + 1}) \in F$.
    
    The other $k$ windows each have the form $v_{i + 1}^mv_i^{k + 1 - m}$ for $1 \le m \le k$. The first candidate of each such window starts with $m$ $v_{i + 1}$'s followed by $v_i$'s, whereas the second starts with $m - 1$ $v_{i + 1}$'s followed by $v_i$'s. Therefore, if $v_{i + 1} < v_i$ (i.e., $(v_i, v_{i + 1}) \in F$), the former candidate will be the window's minimizer whereas if $v_i < v_{i + 1}$, the second candidate will be the minimizer.
    We have $k$ of these windows, all with consecutive starting positions. The minimizer candidates starting at the third through the $(k + 1)$th $v_{i + 1}$ of the block are each \emph{both} the first candidate of some such window and the second candidate of some other such window, and will therefore be a minimizer regardless of the alphabet ordering. This amounts to $k - 1$ minimizers for each arc. Out of these $k$ windows, only the first candidate of the first window and the second candidate of the last window remain. The first candidate of the first window, which starts at the second $v_{i + 1}$ in the block, is a minimizer if and only if $v_{i + 1} < v_i$. The second candidate of the last block is $v_i^{k}$, of which we already know it is a minimizer regardless of the alphabet ordering as it is the start of a new block. Thus, we have one more minimizer if $v_{i + 1} < v_i$ and one more regardless of the ordering.

    Finally, we have a block $v_0^{k + 1}$ at the end of $S$. Like any other block, the first $v_0$ is a minimizer by tie-breaking. The second position in this block, indicating the start of the final length-$k$ fragment of $S$, is not a minimizer because it is not part of any other window.

    In summary, we have:
    \begin{itemize}
        \item for every arc $(v_i, v_{i + 1})\in A$:
        \begin{itemize}
            \item $k$ minimizers regardless of alphabet ordering;
            \item one minimizer if $(v_i, v_{i + 1}) \in F$;
        \end{itemize}
        \item one minimizer at the start of the final block, regardless of alphabet ordering.
    \end{itemize}

All of this adds up to $1 + |A| \cdot k + |F|$ minimizers in total, completing the proof.
\end{proof}

\begin{example}
Let $G(V,A)$ be an Eulerian directed graph with $V=\{\texttt{a},\texttt{b},\texttt{c}\}$ and $A=\{(\texttt{a},\texttt{b}),(\texttt{b},\texttt{c}),(\texttt{c},\texttt{a})\}$. Further let $\texttt{a}<\texttt{b}<\texttt{c}$, $F=\{(\texttt{c},\texttt{a})\}$ , $w=2$, $k=5$, and $w+k-1=6$. For the circuit $\texttt{a},\texttt{b},\texttt{c},\texttt{a}$ ($\texttt{a},\texttt{c},\texttt{b},\texttt{a}$ in reverse) we have $S=\texttt{\textcolor{red}{aaaaaa}\textcolor{red}{c}c\textcolor{red}{cccc}\textcolor{red}{b}b\textcolor{red}{bbbba}aaaaa}$
with $1 + |A| \cdot k + |F|=17$ minimizers colored red.
\end{example}

\begin{proof}[Proof of Theorem~\ref{the:main-w2}]
    We show that the $\MMD$ problem is NP-complete if $w = 2$ by a reduction from \EFASD. Let $(G, \ell')$ be an instance of \EFASD. First, we determine an Eulerian circuit $v_0,\dots,v_{|A|}$ of $G$, which is well-known to be possible in polynomial time. Using this Eulerian circuit, we construct $S$ in accordance with Equation~\ref{eq:construction-w2}. By Lemma~\ref{lem:mwk-w2}, we know that $\Mwk(S, F) = 1 + |A| \cdot k + |F|$. Let $\ell = 1 + |A| \cdot k + \ell'$. Because the function $\Mwk(S, F)$ is linear in $|F|$, we have that $\Mwk(S, F) \le \ell$ if and only if $|F| \le \ell'$. In other words, $G$ contains a feedback arc set of size at most $\ell'$ if and only if there exists some alphabet ordering such that $S$ has at most $\ell$ minimizers. The string $S$ has length $(k + 1) \cdot (|A| + 1)$, which is polynomial in the size of $G$ and can therefore be constructed in polynomial time. 
\end{proof}

\section{Considering the Orderings on \texorpdfstring{$\Sigma^k$}{}}\label{sec:MMk}

Most of the existing approaches for minimizing the minimizers samples consider the space of all orderings on $\Sigma^k$ instead of the ones on $\Sigma$.
Such an approach has the advantage of an easy and efficient construction of the sample by using a rolling hash function $h:\Sigma^k\rightarrow \mathbb{N}$, such as the popular Karp-Rabin fingerprints~\cite{DBLP:journals/ibmrd/KarpR87};
this results in a random ordering on $\Sigma^k$ that usually performs well in practice~\cite{DBLP:journals/jcb/ZhengMK23}.
Let us denote by \MMk the version of \MM that seeks to minimize $|\Mwk(S)|$ by choosing a best ordering on $\Sigma^k$ (instead of a best ordering on $\Sigma$).
It is easy to see that any algorithm solving \MM solves also \MMk with a polynomial number of extra steps:
We use an arbitrary ranking function $\textsf{rank}$ from \emph{the set} of length-$k$ substrings of $S$ to $[1,n-k+1]$. We construct the string $S'$ such that $S'[i]=\textsf{rank}(S[i\dd i+k-1])$, for each $i\in[1,n-k+1]$. 
Let $\Sigma'$ be the set of all letters in $S'$. 
It should be clear that $|\Sigma'| \leq n$ because $S$ has no more than $n$ substrings of length $k$. 
We then solve the \MM problem with input $\Sigma:=\Sigma'$, $S:=S'$, $w:=w$, and $k:=1$.
It is then easy to verify that an optimal solution to \MM for this instance implies an optimal solution to \MMk for the original instance.
We thus conclude that \MM is at least as hard as \MMk; they are clearly equivalent for $k=1$.

\begin{example}
Let $S = \texttt{aacaaacgcta}$, $w = 3$, and $k=3$.
We construct the string $S'=\texttt{235124687}$ over $\Sigma'=[\texttt{1},\texttt{8}]$ and solve
\MM with $w = 3$, $k=1$, and $\Sigma=\Sigma'$.
Assuming $\texttt{1}<\texttt{3}<\texttt{5}<\texttt{6}<\texttt{2}<\texttt{4}<\texttt{7}<\texttt{8}$, $\mathcal{M}_{3,1}(S')=\mathcal{M}_{3,3}(S)=\{2,4,7\}$. The minimizers positions are colored red: $S'=\texttt{2\textcolor{red}{3}5\textcolor{red}{1}24\textcolor{red}{6}87}$. This is one of many best orderings. 
\end{example}

Another advantage of \MMk is that a best ordering on $\Sigma^k$ is at least as good as a best ordering on $\Sigma$ at minimizing the resulting sample. Indeed this is because every ordering on $\Sigma$ implies an ordering on
$\Sigma^k$ but not the reverse. 

Unfortunately, \MMk comes with a major disadvantage. Suppose we had an algorithm solving \MMk (either exactly or with a good approximation ratio or heuristically) and applied it to a string $S$ of length $n$, with parameters $w$ and $k$. Now, in order to compare a query string $Q$ to $S$, the first step would be to compute the minimizers of $Q$, but to ensure local consistency (Property 2), we would need access to the ordering output by the hypothetical algorithm. The size of the ordering is $\cO(\min(|\Sigma|^k,n))$ and storing this defeats the purpose of creating a sketch for $S$.
This is when it might be more appropriate to use \MM instead.

Since \MM is NP-hard for $w \ge 2$ and $k = 1$, \MMONE is NP-hard for $w \ge 2$; hence the following corollary of \cref{cor:main}.
\begin{corollary}\label{coro:k}
    \MMONE is NP-hard if $w \ge 2$.
    \label{cor:MMk}
\end{corollary}

\section{Open Question}\label{sec:fin}

Is \MMk NP-hard for $k > 1$?

\bibliography{references}
\end{document}